\documentclass[runningheads]{llncs}
\usepackage{tikz}
\usetikzlibrary{arrows.meta,positioning}
\usetikzlibrary{backgrounds,calc}
\usetikzlibrary{patterns}
\usepackage[T1]{fontenc}
\usepackage[ruled,linesnumbered,vlined]{algorithm2e}

\newcommand{\ctxedge}{\,\tikz[baseline=-0.5ex]{\draw[blue, dashed, ->, thick] (0,0) -- (0.7,0);}\,}
\newcommand{\fbedge}{\,\tikz[baseline=-0.5ex]{\draw[red, dotted, ->, thick] (0,0) -- (0.7,0);}\,}

\usepackage{graphicx}
\usepackage[table]{xcolor}
\usepackage{array}
\usepackage{booktabs}
\usepackage{listings}
\usepackage{tcolorbox}
\usepackage{tabularx}
\usepackage{amsmath}
\usepackage{amssymb}
\lstdefinelanguage{JavaScript}{
  keywords={function,let,if,return},
  keywordstyle=\bfseries,
  sensitive=true,
  comment=[l]{//},
  morestring=[b]',
  morestring=[b]"
}
\definecolor{hatchgray}{RGB}{210,210,210}

\definecolor{nodegray}{HTML}{F2F2F2}
\definecolor{compgreen}{HTML}{DFF5DF}
\definecolor{ctxblue}{HTML}{3646B7}
\definecolor{fbred}{HTML}{C96573}

\newcommand{\nodelabel}[2]{%
  \begin{tabular}{c}
    \texttt{#1}\\[-0.25em]
    \scriptsize $#2$
  \end{tabular}%
}
\newcommand{\purpmark}{\tikz[baseline=-0.6ex]{\node[circle, inner sep=2pt, fill=purple!70!black] {};}}
\newcommand{\greenmark}{\tikz[baseline=-0.6ex]{\node[circle, inner sep=2pt, fill=green!50!black] {};}}

\begin{document}

\title{Agentic Interpretation: Lattice-Structured Evidence for LLM-Based Program Analysis}
\titlerunning{Agentic Interpretation for Program Analysis}
\author{
  Jacqueline L. Mitchell 
  \and
  Chao Wang 
\institute{
University of Southern California, Los Angeles, California, USA
}
}
\authorrunning{Mitchell and Wang}
\maketitle
\begin{abstract}
Large language models can consult information that fixed static analyzers cannot, such as documentation, current security advisories, version-specific metadata, and informal API contracts.
This makes LLMs a compelling option for program analyses that depend on information beyond the source program, or that are otherwise not amenable to conventional static analyzers.
However, directly asking an LLM for a one-shot whole-program analysis is brittle because it compresses many evidence-dependent judgments into a single opaque answer, rather than exposing which conclusions are supported or disputed and using intermediate findings to guide later, more focused searches.
In this paper, we propose \emph{agentic interpretation}, a framework that brings the discipline of lattice-based static analysis to LLM-driven program reasoning.
At a high level, agentic interpretation decomposes a high-level analysis goal into localized claims, and tracks the LLM's judgment about each claim in a finite-height lattice.  A worklist algorithm governs how claims and their judgments evolve during the analysis.
We introduce a formal model of agentic interpretation, explore the design space it opens, and illustrate the approach with a worked example analyzing code that depends on opaque third-party components.

\end{abstract}

\section{Introduction}

Large language models (LLMs) have become widely used in program analysis and verification. 
The dominant paradigm is \emph{generate-and-check}: an LLM proposes a formal artifact (e.g., an invariant~\cite{DBLP:conf/icml/PeiBSSY23,wei2026quokka}, proof step~\cite{DBLP:conf/sigsoft/FirstRRB23,DBLP:conf/iclr/ChenL0GYLMYD00L25}, or candidate specification~\cite{DBLP:journals/pacmse/EndresFCL24,DBLP:journals/corr/abs-2510-12702}) and an external verifier decides whether the artifact is acceptable.
LLMs have also been used as heuristic oracles to guide program analysis, where the LLM can only make decisions that do not affect soundness~\cite{DBLP:conf/iclr/0001BN24,stechly2025on,DBLP:journals/pacmpl/CaiHSLLSD25,DBLP:journals/corr/abs-2508-14532,DBLP:conf/iclr/Li0N25,DBLP:journals/corr/abs-2412-14399}. 
These paradigms are compelling because the correctness argument is independent of the LLM.
However, these paradigms work precisely because the LLM's contribution can be formally validated. 
Many important program analysis questions do not admit such validation, either because designing a suitable verifier is prohibitively difficult or because the relevant evidence resists formalization. 
An analyst may want to know whether an external parser is safe, whether a recently disclosed vulnerability affects the version in use, or whether a call respects an undocumented API convention.
These questions concern program structure, possible flows, and component interactions, but the evidence needed to answer them is not fully contained in the source code and is often not available as a formal specification.
Instead, it may come from documentation, package metadata, current security advisories, informal specifications, engineering conventions, or a mixture of these. We refer to such problems as \emph{evidence-dependent program analyses}, and Figure~\ref{fig:use-cases} presents representative members of this category.
LLM agents are a natural fit for these analyses because they can retrieve, interpret, and synthesize exactly the kinds of information that fixed analyzers cannot access. 
Unlike a conventional analyzer, whose mechanics are defined at design time, an LLM can incorporate ecosystem knowledge that changes after the analyzer is built and reason about real-time information and qualitative aspects of the code.

\begin{figure}[t]
    \centering
    \scalebox{0.8}{    
    \includegraphics[width=\linewidth]{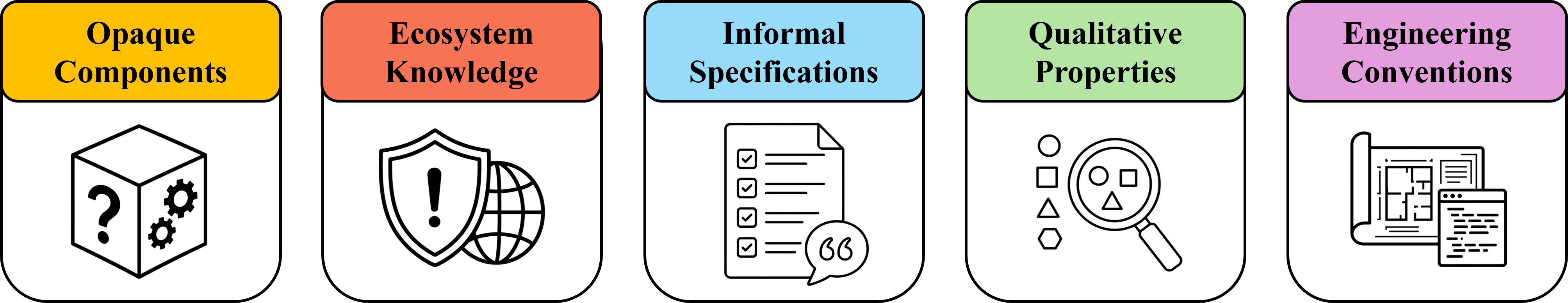}
    }
    \caption{Five Target Categories for Agentic Interpretation}
    \label{fig:use-cases}
\end{figure}

Asking an LLM to do the analysis directly is brittle for three reasons.
First, empirical studies show that model performance can depend on where relevant information appears in the prompt, and can degrade as context length and task complexity increase~\cite{DBLP:journals/tacl/LiuLHPBPL24,DBLP:conf/emnlp/DuTRRBGWSHP25,DBLP:conf/acl/TianLFDLQWCZWLW25,DBLP:conf/acl/YuJLWL000Q25}.
This means that a version-specific advisory or an informal library contract may be supplied to the model but not used reliably.
Second, a monolithic answer has no persistent state. 
Even when it contains plausible local observations about the code, it may not record which claims were supported, refuted, left open, or supported only by weakly applicable evidence. 
This is especially problematic because LLM-generated reasoning can be self-contradictory or fail to support its own conclusions~\cite{DBLP:conf/nips/TurpinMPB23,DBLP:conf/acl/LiC0XLJLZ25,DBLP:conf/emnlp/TutekCMB25,DBLP:conf/acl/0002Y025}.
Third, a one-shot answer has no explicit dependency structure.
Later findings may reveal that an earlier question was scoped incorrectly or that an upstream component should be reexamined under a more focused hypothesis, but an unstructured answer has no mechanism for such feedback.

This paper proposes \emph{agentic interpretation} as a middle ground between ad hoc LLM prompting and a fully specified static analysis.
The framework decomposes a high-level analysis goal and a graph-based representation of a program into localized claims anchored to program points, component boundaries, and auxiliary analysis nodes.
Each claim is assessed by an LLM agent under controlled context.  The resulting judgment is recorded in a finite-height assessment lattice, and a worklist-based harness governs how claims and judgments evolve and how information is propagated between nodes.
The lattice is not an abstract domain over concrete program states in the classical sense.
Instead, it is a structured domain over evidence states: a lattice element records the current evidential status of a claim (e.g., whether the claim has weak or strong support).

The goal of agentic interpretation is to achieve a separation of concerns. The role of the LLM agent is to provide evidence by retrieving and interpreting information that a conventional static analyzer would not naturally consult. The role of the lattice-based worklist harness is to ensure that the analysis proceeds in a controlled and principled way that is bounded and auditable.
In summary, this paper makes three contributions: (1)~it proposes the framework of agentic interpretation and identifies a suitable problem domain for it, (2)~formalizes the core model, and (3)~characterizes and discusses the design space of agentic interpretation.
Moreover, this paper describes work in progress.  A practical implementation and experimental evaluation of the system are left to future work.
\section{Example: Opaque Component Security Review}

This section illustrates agentic interpretation on a security review task involving a small program with opaque third-party components.
The program (Figure~\ref{fig:running-ex}) is a request handler that parses untrusted input using an external JSON library, verifies a cryptographic signature using an external signing library, and processes the resulting payload if the check succeeds.
The source code of the handler and the processor is available to the analyzer, but the source code of the parser and the verifier is \textbf{not}. 
(We note that the library names and security advisories used in the example are synthetic.)
\begin{figure}[t]
\begin{lstlisting}[language=JavaScript, numbers=left, numberstyle=\tiny, basicstyle=\ttfamily\scriptsize, columns=flexible, escapeinside={(*}{*)}]
// Source available
function handleRequest(rawInput) {
  let payload = parseJSON(rawInput);     // json-fast
  if (verifySignature(payload)) {        // sign-lib
    return processPayload(payload);
  }
  return reject(rawInput);
}

// Source available
function processPayload(payload) {
  const amount = payload.amount;
  const recipient = payload.recipient;
  const accountId = payload.accountId;
  return submitTransfer(accountId, recipient, amount);
}

// Third-party sources unavailable:
// parseJSON(input): json-fast v3.2.1
// verifySignature(payload): sign-lib v2.0.4
// declared verification scope supplied by external configuration metadata
\end{lstlisting}
\caption{A handler with two opaque third-party components. The source of \texttt{parseJSON} and \texttt{verifySignature} is unavailable to the analyzer; the source of \texttt{handleRequest} and \texttt{processPayload} is available.}
\label{fig:running-ex}
\end{figure}
Suppose that a human analyst provides the high-level security goal to check:
\begin{quote}
    \textbf{Goal: } \texttt{handleRequest} should either reject the input or call \texttt{processPayload} only with transfer-relevant fields that are signed, schema-valid, own properties of \texttt{payload}, and protected against parser-induced object-shape attacks.
\end{quote}
Here, ``schema-valid'' means valid with respect to an analyst-supplied transfer-payload contract that identifies \texttt{amount}, \texttt{recipient}, and \texttt{accountId} as transfer-relevant fields and specifies their expected types and value constraints. 
This is an evidence-dependent analysis problem because the goal depends on the ecosystem state of \texttt{json-fast}~v3.2.1 (e.g., documentation, advisories), on the informal contract of \texttt{verifySignature}, and on application-specific relevance determined by downstream field use.

\subsection{Evaluation Graph and Claims}

\begin{figure}[t]
\centering
\begin{minipage}{\linewidth}
\centering
\scalebox{0.7}{
\begin{tikzpicture}[
  flow/.style={-{Stealth[length=2.2mm]}, line width=0.55pt},
  ctx/.style={-{Stealth[length=2.2mm]}, dashed, line width=0.65pt, draw=ctxblue},
  fb/.style={-{Stealth[length=2.2mm]}, dotted, line width=0.85pt, draw=fbred},
  prog/.style={
    draw,
    rounded corners=3pt,
    fill=nodegray,
    minimum width=2.75cm,
    minimum height=0.72cm,
    inner sep=2pt,
    align=center
  },
  opaque/.style={
    prog,
    postaction={pattern=north east lines, pattern color=hatchgray}
  },
  comp/.style={
    draw,
    rounded corners=3pt,
    fill=compgreen,
    minimum width=3.05cm,
    minimum height=0.86cm,
    inner sep=2pt,
    align=center
  },
  branch/.style={font=\scriptsize, fill=white, inner sep=1pt}
]
  \node[prog] (raw) at (0,0)
    {\nodelabel{rawInput}{n_0}};
  \node[opaque] (parse) at (0,-1.35)
    {\nodelabel{parseJSON}{n_1 : c_P}};
  \node[opaque] (verify) at (0,-2.70)
    {\nodelabel{verifySignature}{n_2 : c_{V_\ell}, c_{V_s}}};
  \node[prog] (process) at (-3.05,-4.35)
    {\nodelabel{processPayload}{n_3 : c_U}};
  \node[prog] (reject) at (3.05,-4.35)
    {\nodelabel{reject}{n_4 : c_R}};
  \node[prog] (exit) at (0,-5.55)
    {\nodelabel{exit}{n_5 : c_G}};
  \node[comp] (compose) at (4.15,-2.70)
    {\nodelabel{composition}{n_C : c_C}};
  % Program-control edges
  \draw[flow] (raw) -- (parse);
  \draw[flow] (parse) -- (verify);
  \draw[flow]
    (verify.south west)
    -- node[branch, pos=0.35, above left] {true}
    (process.north east);
  \draw[flow]
    (verify.south east)
    -- node[branch, pos=0.35, above right] {false}
    (reject.north west);
  \draw[flow] (process.south east) -- (exit.north west);
  \draw[flow] (reject.south west) -- (exit.north east);
  % Context edges into composition
  \draw[ctx]
    (parse.east)
    to[out=4,in=172]
    (compose.west);
  \draw[ctx]
    (verify.east)
    -- (compose.west);
  \draw[ctx]
    (process.east)
    to[out=12,in=220]
    (compose.south west);
  \coordinate (cgA) at ($(compose.south east)+(0.65,-0.10)$);
  \coordinate (cgB) at ($(compose.south east)+(1.10,-2.35)$);
  \coordinate (cgC) at ($(exit.east)+(0.95,0.20)$);
  \coordinate (rgA) at ($(reject.south east)+(0.35,-0.10)$);
  \coordinate (rgB) at ($(reject.south east)+(0.85,-1.10)$);
  \coordinate (rgC) at ($(exit.east)+(0.55,-0.20)$);
  % Context edges into global goal at exit
  \draw[ctx, rounded corners=6pt]
    (compose.south east) -- (cgA) -- (cgB) -- (cgC) -- (exit.east);
  \draw[ctx, rounded corners=6pt]
    (reject.south east) -- (rgA) -- (rgB) -- (rgC) -- (exit.east);
  % Feedback edges
  \draw[fb]
    (compose.west)
    to[out=158,in=12]
    (parse.east);
  \draw[fb]
    (compose.west)
    to[out=202,in=0]
    (verify.east);
\end{tikzpicture}
}
\end{minipage}

\vspace{0.8em}

\begin{minipage}{\linewidth}
\centering
\renewcommand{\arraystretch}{1.15}
\scalebox{0.7}{
\begin{tabularx}{\linewidth}{@{}c c >{\raggedright\arraybackslash}X@{}}
\toprule
Node & Claim & Meaning \\
\midrule
$n_1$
  & $c_P$
  & \texttt{parseJSON} does not introduce attacker-controlled inherited
    properties or other object-shape corruption relevant to consumed fields. \\

$n_2$
  & $c_{V_\ell}$
  & \texttt{verifySignature} correctly validates signatures over the fields
    it is asked to verify. \\

$n_2$
  & $c_{V_s}$
  & The declared verification scope covers all security-relevant fields
    consumed downstream and constrains the relevant field provenance. \\

$n_3$
  & $c_U$
  & \texttt{processPayload} independently checks that consumed fields are
    own properties with expected types and schema. \\

$n_4$
  & $c_R$
  & Failed signature verification leads to rejection of the input. \\

$n_C$
  & $c_C$
  & The true-branch composition processes only signed, schema-valid, and
    structurally safe security-relevant fields. \\

$n_5$
  & $c_G$
  & \texttt{handleRequest} either rejects the input or calls \texttt{processPayload} only with transfer-relevant fields that are signed, schema-valid, own properties of payload, and protected against parser-induced object-shape attacks. \\
\bottomrule
\end{tabularx}
}
\end{minipage}
\caption{
Evaluation graph and claims for the motivating example.
Context edges (\protect\ctxedge) flow from component claims to the cross-component claim $(c_C)$ and from $c_C$ and $c_R$ to $c_G$.
Feedback edges (\protect\fbedge) let a refuted cross-component claim focus later searches about opaque components.
The table summarizes the claims associated with each node.
Hashed nodes indicate that the nodes represent opaque components.
}
\label{fig:motivating-example}
\end{figure}
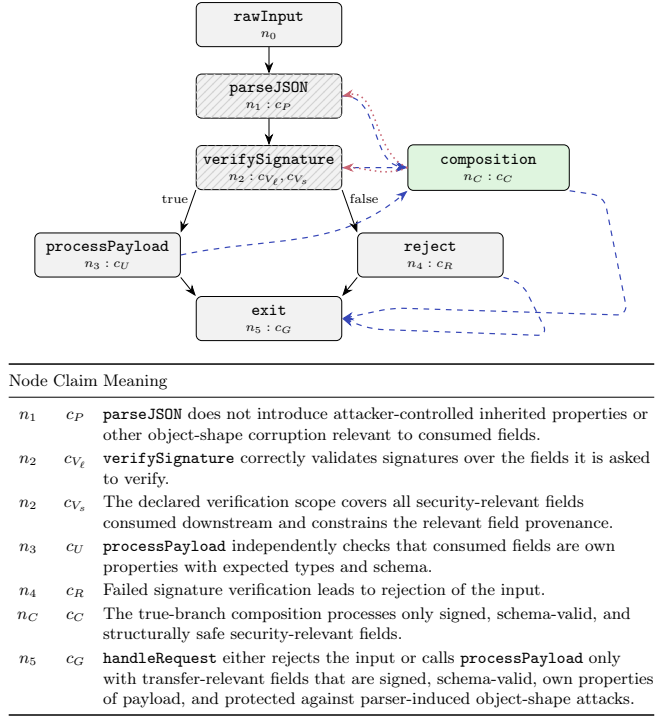

The subgraph with grey nodes and solid edges represents the program-derived graph for the code in Figure~\ref{fig:running-ex}.
Each grey node represents a program point or a component boundary (e.g., $n_3$ represents \texttt{processPayload} and $n_5$ represents the exit point), and each black edge represents control flow.
This graph captures the structure of the program, but does not capture the dependencies that arise when evaluating the analyst's security goal.
In particular, the goal depends on how multiple components interact (e.g., the true branch represented by the path from $n_1$ to $n_3$), and downstream findings about that interaction may need to redirect the investigation of upstream components.
The evaluation graph extends the program-derived graph with four types of additions to represent these dependencies:

\paragraph{Claims.} The evaluation graph associates each node with one or more claims: propositions that the framework asks the LLM agent to evaluate by finding evidence.
We note that the assessment of the verifier is split into two claims, where $c_{V_\ell}$ concerns local cryptographic correctness, while $c_{V_s}$ concerns whether the declared scope covers all security-relevant fields consumed downstream.
This split reflects the fact that a verifier can be correct within its own contract while providing insufficient coverage for the handler's security goal.
The table in Figure~\ref{fig:motivating-example} describes all of the claims in detail.

\paragraph{Auxiliary Nodes.} Some claims span multiple components and cannot easily be anchored to a specific node in the program-derived graph.  The security goal requires that the parser, verifier, and processor \emph{jointly} ensure that only signed, schema-valid, and structurally safe fields reach the processor.
An auxiliary node $n_C$ (the green node in Figure~\ref{fig:motivating-example}) is introduced to represent this interaction and carries a cross-component claim spanning the true-branch components.

\paragraph{Context Edges.} When the LLM agent evaluates a claim at one node, it may need the information and context from other nodes.  Context edges ($\ctxedge$) make this dependency explicit.
In Figure~\ref{fig:motivating-example}, $n_1 \ctxedge n_C$, $n_2 \ctxedge n_C$, and $n_3 \ctxedge n_C$ ensure that when $c_C$ at $n_C$ is evaluated, the agent receives the relevant information from the parser, verifier, and processor claims.
Similarly, edges $n_4 \ctxedge n_5$ and $n_C \ctxedge n_5$ supply the goal claim at the exit node with the context from the rejection and composition claims.
Context edges flow forward in the sense that upstream findings inform downstream evaluation.

\paragraph{Feedback Edges} Later findings may reveal that an earlier question was scoped too broadly or that an upstream component should be reexamined under a more focused hypothesis.
Feedback edges ($\fbedge$) propagate this signal in the reverse direction.  When a downstream assessment changes, the upstream node is scheduled to be investigated further.
The edges $n_C \fbedge n_1$ and $n_C \fbedge n_2$ allow a refuted cross-component claim to trigger retargeted queries about the opaque parser and verifier (e.g., by narrowing a broad documentation review toward a specific class of vulnerability).
There is no feedback edge from $n_C$ to $n_3$, because the processor's source code is available, so the lack of validation can be assessed directly instead of searching for external evidence.
Here feedback edges focus on the opaque components only, as the source code of other nodes can be directly inspected.

Both auxiliary and program-derived nodes participate in the worklist computation, while only the feedback and context edges drive the worklist computation.

\subsection{A Graded Assessment Domain}
\label{sec:graded-domain}
For the running example, each claim is assessed in the product lattice:
\[
A_{\textsc{Graded}} = \{\bot, w, s\} \times \{\bot, w, s\},
\] where $\bot \sqsubseteq w \sqsubseteq s$.  The full lattice is depicted in Figure~\ref{fig:graded-lattice}. Here, the strength of evidence is classified as absent ($\bot$), weak ($w$), or strong ($s$).
The first coordinate represents supporting evidence and the second represents refuting evidence, and the order is pointwise.
For example, $\langle w, s \rangle$ means that weak supporting evidence and strong refuting evidence for a given claim has been observed by the LLM agent.
The lattice value is deliberately compact, and we assume that the worklist harness separately stores evidence records (e.g., such as source observations, documentation excerpts, advisory identifiers).
This separation lets the worklist use a small ordered domain while preserving enough evidence to audit the result and construct later prompts.

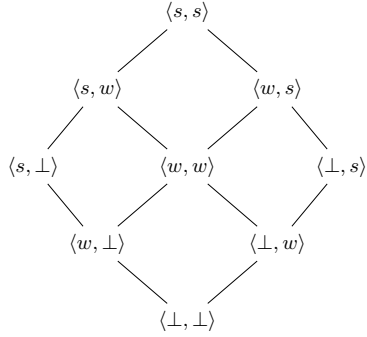
\begin{figure}[t]
\centering
\scalebox{0.85}{
\begin{tikzpicture}[
  every node/.style={inner sep=3pt}
]
  \node (bot)  at (0,0)      {$\langle\bot,\bot\rangle$};
  \node (w0)   at (-1.4,1.2) {$\langle w,\bot\rangle$};
  \node (0w)   at (1.4,1.2)  {$\langle\bot,w\rangle$};
  \node (s0)   at (-2.4,2.4) {$\langle s,\bot\rangle$};
  \node (ww)   at (0,2.4)    {$\langle w,w\rangle$};
  \node (0s)   at (2.4,2.4)  {$\langle\bot,s\rangle$};
  \node (sw)   at (-1.4,3.6) {$\langle s,w\rangle$};
  \node (ws)   at (1.4,3.6)  {$\langle w,s\rangle$};
  \node (top)  at (0,4.8)    {$\langle s,s\rangle$};

  \draw (bot) -- (w0);
  \draw (bot) -- (0w);
  \draw (w0) -- (s0);
  \draw (w0) -- (ww);
  \draw (0w) -- (ww);
  \draw (0w) -- (0s);
  \draw (s0) -- (sw);
  \draw (ww) -- (sw);
  \draw (ww) -- (ws);
  \draw (0s) -- (ws);
  \draw (sw) -- (top);
  \draw (ws) -- (top);
\end{tikzpicture}
}
\caption{The $A_{\textsc{Graded}}$ lattice. Each component takes values from $\{\bot, w, s\}$, ordered $\bot \sqsubseteq w \sqsubseteq s$. Joins are pointwise maxima. A state such as $\langle w,s\rangle$ records weak support and strong refutation; it does not mean the claim is both true and false.}
\label{fig:graded-lattice}
\end{figure}

\subsection{Contextualization, Queries, and Evidence Records}
\label{sec:context-queries-evidence}

When the worklist harness processes a node, it constructs local context from the source code of neighboring nodes, the goal, the current assessment state, and information from predecessors via extended edges (i.e., the feedback and context edges).
For example, for the cross-component node $n_C$, the context includes the source of \texttt{handleRequest} and \texttt{processPayload}, the goal, and the current assessments and evidence for $c_P$ $(n_1)$, $c_{V_\ell}$ $(n_2)$, $c_{V_s}$ $(n_2)$, and $c_U$ $(n_3)$.
A bilateral evaluation for $c_C$ can be implemented by the following two subqueries.  
A key point is that the structure of the subqueries reflects the structure of the assessment lattice.  The result of each of these subqueries will determine the value of the corresponding coordinate in the lattice element for $c_C$.

\begin{tcolorbox}[
  title={\small\textbf{Bilateral evaluation query for $c_C@n_C$}},
  colback=gray!3,
  colframe=gray!50!black,
  fontupper=\small,
  boxrule=0.5pt,
  arc=2pt,
  left=4pt, right=4pt, top=2pt, bottom=2pt
]
Given the context and claim below, evaluate the claim using the $A_{\textsc{Graded}}$ evidence lattice.

\medskip
\textbf{Claim:} The true-branch composition processes only signed, schema-valid, and structurally safe security-relevant fields.

\medskip
\textbf{Context:} The source of \texttt{handleRequest} and \texttt{processPayload}, the goal $\mathbf{Goal}$, and the current assessments (and associated evidence) of $c_P$, $c_{V\ell}$, $c_{Vs}$, and $c_U$.

\medskip
Find evidence that \textcolor{green!40!black}{\textbf{supports}} the claim. In particular, check whether the parser, verifier, and processor jointly establish signature coverage, own-property status, schema validity, and protection against object-shape attacks. Classify the support level as $\bot$ (absent), $w$ (weak), or $s$ (strong).

\medskip
Find evidence that \textcolor{red!40!black}{\textbf{refutes}} the claim. In particular, identify any downstream field assumption that is not established by the parser and verifier jointly, or any component behavior that could undermine signature coverage, field provenance, schema validity, or structural safety. Classify the refutation level as $\bot$ (absent), $w$ (weak), or $s$ (strong).

\medskip
Use \emph{w} when the evidence is indirect, incomplete, generic, omission-based, or not clearly applicable to the specific version, code path, API contract, or security obligation under analysis. Use \emph{s} when the evidence is direct, specific, applicable, and sufficient to support or refute the claim in the given context.

\medskip
Return the support level, supporting evidence records, refutation level, refuting evidence records, and a brief rationale.
\end{tcolorbox}

\subsection{Evidence Stabilization: One Worklist Trace}

Table~\ref{tab:trace} summarizes one possible run of the worklist harness while analyzing the code in Figure~\ref{fig:running-ex}.
Each step in Table~\ref{tab:trace} processes one node from the worklist.
The harness assembles a prompt context for the node (including source code, the high-level goal, and the current assessments at extended predecessors), queries the LLM agent, and joins the returned assessment with the node's current value. 
If the join produces a strict increase, the node's extended successors are enqueued.
A dot ($\cdot$) indicates that a claim's assessment is unchanged from the previous row.
The \emph{Join} column shows the join expression for the claim that changed, and $W$ lists the worklist contents after the step.
The initial worklist is $W_0 = \{n_1, n_2, n_3, n_4, n_C, n_5\}$.
The auxiliary node~$n_C$ and the exit node~$n_5$ are included despite the fact that their claims depend on context from upstream nodes to ensure that every claim is evaluated at least once.
We now discuss the reasoning behind each step seen in Table~\ref{tab:trace}, grouping the steps into four phases for presentation.
 
\newcommand{\lightrule}{\arrayrulecolor{black!25}\midrule\arrayrulecolor{black}}
\begin{table}[t]
\centering
\caption{State evolution for one run.  Columns $c_P$ through
$c_G$ show the assessment after each step; a dot indicates no change.
The \emph{Join} column gives the join expression for the claim that
changed, and $W$ lists the worklist contents after the step. Here, $\bot^2$ is shorthand
for $\langle \bot, \bot \rangle$ in the interest of space.}
\label{tab:trace}
 
\resizebox{\textwidth}{!}{%
\renewcommand{\arraystretch}{1.15}
\begin{tabular}{@{}cl l ccccccc l l@{}}
\toprule
Step & Node & Action
  & $c_P$ & $c_{V\ell}$ & $c_{Vs}$ & $c_U$ & $c_R$ & $c_C$ & $c_G$
  & Join & $W$ after step \\
\midrule
0  & ---   & init
   & $\bot^2$ & $\bot^2$ & $\bot^2$ & $\bot^2$ & $\bot^2$ & $\bot^2$ & $\bot^2$
   & ---
   & $\{n_1,n_2,n_3,n_4, n_C, n_5\}$ \\
\midrule
\multicolumn{12}{@{}l}{\textit{Phase I: Initial Component Assessment}} \\
\addlinespace[2pt]
1  & $n_1$ & broad parser review
   & $\langle w,\bot\rangle$ & $\cdot$ & $\cdot$ & $\cdot$ & $\cdot$ & $\cdot$ & $\cdot$
   & $\bot^2 \sqcup \langle w,\bot\rangle$
   & $\{n_2,n_3,n_4,n_C,n_5\}$ \\
\lightrule
2  & $n_2$ & broad verifier review
   & $\cdot$ & $\langle s,\bot\rangle$ & $\langle\bot,w\rangle$ & $\cdot$ & $\cdot$ & $\cdot$ & $\cdot$
   & $\begin{array}{@{}l@{}} c_{V\ell}\colon \bot^2 \sqcup \langle s,\bot\rangle \\ c_{Vs}\colon \bot^2 \sqcup \langle\bot,w\rangle \end{array}$
   & $\{n_3,n_4,n_C, n_5\}$ \\
\lightrule
3  & $n_3$ & inspect processor src.
   & $\cdot$ & $\cdot$ & $\cdot$ & $\langle\bot,s\rangle$ & $\cdot$ & $\cdot$ & $\cdot$
   & $\bot^2 \sqcup \langle\bot,s\rangle$
   & $\{n_4,n_C, n_5\}$ \\
\lightrule
4  & $n_4$ & inspect rejection branch
   & $\cdot$ & $\cdot$ & $\cdot$ & $\cdot$ & $\langle s,\bot\rangle$ & $\cdot$ & $\cdot$
   & $\bot^2 \sqcup \langle s,\bot\rangle$
   & $\{n_C,n_5\}$ \\
\midrule
\multicolumn{12}{@{}l}{\textit{Phase II: First Composition}} \\
\addlinespace[2pt]
5  & $n_C$ & compose current evidence
   & $\cdot$ & $\cdot$ & $\cdot$ & $\cdot$ & $\cdot$ & $\langle\bot,w\rangle$ & $\cdot$
   & $\bot^2 \sqcup \langle\bot,w\rangle$
   & $\{n_5,n_1,n_2\}$ \\
\lightrule
6  & $n_5$ & compose whole-goal
   & $\cdot$ & $\cdot$ & $\cdot$ & $\cdot$ & $\cdot$ & $\cdot$ & $\langle\bot,w\rangle$
   & $\bot^2 \sqcup \langle\bot,w\rangle$
   & $\{n_1,n_2\}$ \\
\midrule
\multicolumn{12}{@{}l}{\textit{Phase III: Feedback-Driven Refinement}} \\
\addlinespace[2pt]
7  & $n_1$ & targeted parser search
   & $\langle w,s\rangle$ & $\cdot$ & $\cdot$ & $\cdot$ & $\cdot$ & $\cdot$ & $\cdot$
   & $\langle w,\bot\rangle \sqcup \langle\bot,s\rangle$
   & $\{n_2,n_C\}$ \\
\lightrule
8  & $n_2$ & re-check verifier scope
   & $\cdot$ & $\cdot$ & $\langle\bot,s\rangle$ & $\cdot$ & $\cdot$ & $\cdot$ & $\cdot$
   & $\langle\bot,w\rangle \sqcup \langle\bot,s\rangle$
   & $\{n_C\}$ \\
\midrule
\multicolumn{12}{@{}l}{\textit{Phase IV: Final Composition \& Stabilization}} \\
\addlinespace[2pt]
9  & $n_C$ & compose revised context
   & $\cdot$ & $\cdot$ & $\cdot$ & $\cdot$ & $\cdot$ & $\langle\bot,s\rangle$ & $\cdot$
   & $\langle\bot,w\rangle \sqcup \langle\bot,s\rangle$
   & $\{n_5,n_1,n_2\}$ \\
\lightrule
10 & $n_5$ & compose revised goal
   & $\cdot$ & $\cdot$ & $\cdot$ & $\cdot$ & $\cdot$ & $\cdot$ & $\langle\bot,s\rangle$
   & $\langle\bot,w\rangle \sqcup \langle\bot,s\rangle$
   & $\{n_1,n_2\}$ \\
\lightrule
11 & $n_1$ & reprocess (no change)
   & $\cdot$ & $\cdot$ & $\cdot$ & $\cdot$ & $\cdot$ & $\cdot$ & $\cdot$
   & $\langle w,s\rangle \sqcup \langle\bot,s\rangle = \langle w,s\rangle$
   & $\{n_2\}$ \\
\lightrule
12 & $n_2$ & reprocess (no change)
   & $\cdot$ & $\cdot$ & $\cdot$ & $\cdot$ & $\cdot$ & $\cdot$ & $\cdot$
   & $\begin{array}{@{}l@{}} c_{V\ell}\colon \langle s,\bot\rangle \sqcup \langle s,\bot\rangle = \langle s,\bot\rangle \\ c_{Vs}\colon \langle\bot,s\rangle \sqcup \langle\bot,s\rangle = \langle\bot,s\rangle \end{array}$
   & $\emptyset$ \\
\bottomrule
\end{tabular}%
}
\end{table}
 
\paragraph{Phase~I: Initial Component Assessment (Steps 1--4).}
 
In this phase, each component node is evaluated independently, before
any cross-component context is available.
 
\smallskip\noindent\textbf{Step~1} ($c_P@n_1$).  The agent reviews the documentation for \texttt{json-fast}~v3.2.1, which describes standard JSON parsing with robust error handling.  
Without specific guidance on what threat to search for, a broad advisory search does not surface version-specific vulnerabilities.
The agent records weak support for~$c_P$, and would add $n_C$ to the worklist if it was not present, because the assessment of $c_P$ has changed.
 
\smallskip\noindent\textbf{Step~2} ($c_{V\ell}$, $c_{Vs}$ at~$n_2$). 
The agent checks the documentation for \texttt{verifySignature} from \texttt{sign-lib}, which indicates that the verifier checks signatures over a declared set of fields, which the agent treats as strong support for local correctness~($c_{V\ell}$).
For verification scope~($c_{Vs}$), however, the same documentation is \emph{silent} on inherited properties, schema conformance, and fields outside the declared set.
This is an argument by omission: the documentation's silence is treated as evidence against scope coverage, but is weaker than a positive demonstration of a gap, since the documentation may simply not mention a property that is in fact enforced.
The agent accordingly records only weak refutation for~$c_{Vs}$.
 
\smallskip\noindent\textbf{Step~3} ($c_U@n_3$). The agent inspects the source code of \texttt{processPayload} and observes that it accesses payload fields directly (e.g., \texttt{payload.amount}) without performing own-property, type, or schema checks.
This is a direct code observation which yields strong refutation of~$c_U$.
 
\smallskip\noindent\textbf{Step~4} ($c_R@n_4$).  The agent inspects the source and confirms that the false branch of the \texttt{verifySignature} conditional calls \texttt{reject(rawInput)}, yielding strong support for~$c_R$.
This would add~$n_5$ to the worklist, if it was not present.
 
\paragraph{Phase~II: First Composition (Steps 5--6).}
 
With the component assessments in place, the framework now evaluates the cross-component $(c_C)$ and global goal $(c_G)$ claims for the first time.
 
\smallskip\noindent\textbf{Step~5} ($c_C@n_C$).  The agent composes the current true-branch evidence.
The context flowing into~$n_C$ includes
$c_P$ at~$\langle w, \bot\rangle$ (parser weakly supported), $c_{V\ell}$ at~$\langle s, \bot\rangle$ (local correctness strongly
supported), $c_{Vs}$ at~$\langle\bot, w\rangle$ (scope weakly refuted), and $c_U$ at~$\langle\bot, s\rangle$ (processor performs no validation).
These assessments reveal a gap: the processor does not validate consumed fields, and the verifier's documented contract is narrower than the handler's security goal.
However, the agent records only \emph{weak} refutation of~$c_C$ at this stage.
The composition claim asks whether the components \emph{jointly} establish the required security properties, and the upstream picture is still mixed: the parser claim has not yet been refuted, the verifier scope refutation rests on an omission argument rather than a positive finding, and it remains possible that properties enforced by the parser and verifier together compensate for the processor's lack of checks.
No single component has yet been shown to positively fail in a way that closes off this possibility.  Since $c_C$~has changed, the feedback edges $n_C \fbedge n_1$ and $n_C \fbedge n_2$ add the parser and verifier nodes to the worklist for reinvestigation.
 
\smallskip\noindent\textbf{Step~6} ($c_G@n_5$).  The agent evaluates the whole-handler goal.  The rejection branch is strongly supported ($c_R$), but the true-branch composition is weakly refuted ($c_C$).
Because the goal requires \emph{every} modeled execution to satisfy the security conditions, support for the rejection branch alone is insufficient.
The agent records weak refutation of~$c_G$.

\paragraph{Phase~III: Feedback-Driven Refinement (Steps 7--8).}
 
This phase illustrates the framework's advantage over a single-pass analysis.
The composition $(c_C)$'s weak refutation, propagated back to~$n_1$ and~$n_2$ via feedback edges, redirects the upstream investigation from broad documentation review toward targeted threat search.
 
\smallskip\noindent\textbf{Step~7} ($c_P@n_1$, revisited).  The feedback context from~$n_C$ now includes the evidence behind the
composition's refutation: the processor performs no validation, and the verifier's scope does not cover schema conformance or structural safety.
This context reveals that object-shape integrity at the parser is \emph{load-bearing}. If the parser produces an object with inherited or injected properties, no downstream component will catch it.
The agent accordingly narrows its search to prototype pollution and object-shape vulnerabilities in \texttt{json-fast}~v3.2.1, and surfaces a security advisory reporting a pollution bug in nested \texttt{\_\_proto\_\_} handling.
This targeted evidence yields strong refutation of~$c_P$.
 
\smallskip\noindent\textbf{Step~8} ($c_{Vs}@n_2$, revisited).  With the same feedback context as \textbf{Step 7}, the agent re-examines the verifier's scope.
The local correctness of~$c_{V\ell}$ remains supported, but the evidence now \emph{positively} establishes that the verifier's documented scope does not cover schema validity, own-property status, or structural safety for the fields consumed by \texttt{processPayload} after having checked \texttt{sign-lib}'s API contract.
This upgrades~$c_{Vs}$ from weak to strong refutation.

\paragraph{Phase~IV: Final Composition and Stabilization (Steps 9--12).}
 
The strengthened upstream evidence propagates forward through the composition and goal claims, and then the state stabilizes.
 
\smallskip\noindent\textbf{Step~9} ($c_C@n_C$, revised).  The parser now has strong refuting evidence (a concrete advisory), and the verifier
scope has strong refuting evidence (a positive demonstration of insufficient coverage).
The agent determines that this is sufficient to strongly refute the true-branch composition claim.
 
\smallskip\noindent\textbf{Step~10} ($c_G@n_5$, revised).  With the true-branch composition now strongly refuted and the rejection branch still strongly supported, the agent concludes that the global goal has strong refuting evidence: \texttt{handleRequest} can reach \texttt{processPayload} under conditions that do not satisfy the stated security properties.
 
\smallskip\noindent\textbf{Steps~11--12} ($n_1$, $n_2$, reprocessed). Both nodes are revisited because the feedback edges from~$n_C$ added
them to the worklist in Step~9.  However, the feedback context from~$n_C$ has not changed since Step~9. $c_C$ remains at~$\langle\bot,s\rangle$ with the same evidence, so no new evidence is surfaced and all joins are absorbed.
The worklist empties and the analysis terminates. \\

\noindent
An analyst can now check $c_G$ to see what the analysis has determined about the high-level goal and inspect the evidence.
It is assumed that, throughout this process, the relevant evidence collected by the LLM is recorded at each node for each claim, but that the evidence does not directly drive worklist propagation (only changes in assessment do).
\section{The Core Model}
\label{sec:core-model}
This section formalizes agentic interpretation.
The model separates the evaluation graph, the claims and evidence recorded at each node, the assessment lattice, and the worklist harness.
Figure~\ref{fig:agent-interp-flow} summarizes the overall workflow.

\begin{figure}[t]
\centering
\includegraphics[width=0.8\linewidth]{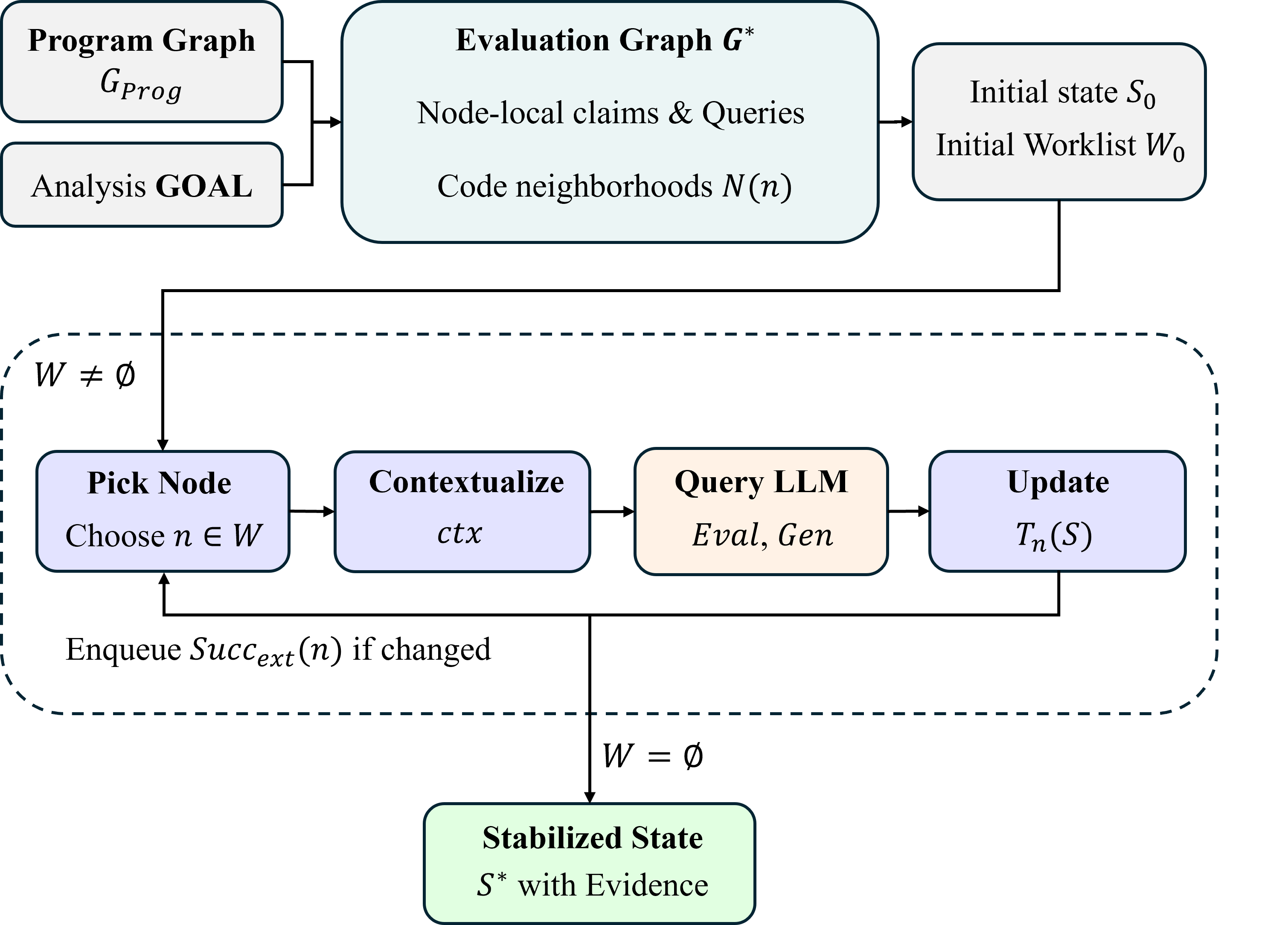}
\caption{
Workflow of agentic interpretation.}
\label{fig:agent-interp-flow}
\end{figure}

\subsection{The Evaluation Graph}

The central data structure of agentic interpretation is the evaluation graph, which extends a program-derived graph with additional edges and nodes that capture analysis-level dependencies.

\begin{definition}[Program-derived graph]
A program-derived graph is a finite directed graph
\[
G_{\mathrm{prog}} = (V_{\mathrm{prog}}, E_{\mathrm{prog}})
\]
where \(E_{\mathrm{prog}} \subseteq V_{\mathrm{prog}} \times V_{\mathrm{prog}}\).  Nodes in \(V_{\mathrm{prog}}\) represent program points or program regions, and edges represent a chosen structural relation, such as control flow, call flow, or data flow.
\end{definition}

The evaluation graph augments the program-derived graph with auxiliary nodes, context edges, and feedback edges to capture analysis-level dependencies absent from the program structure. 
A code neighborhood function is also added to determine which program regions (i.e., source code associated with nodes in $V_{prog}$) are visible when processing a given node.
\begin{definition}[Evaluation graph]
Given $G_{\mathit{prog}}=(V_{\mathit{prog}},\,E_{\mathit{prog}})$,
an \emph{evaluation graph} is a tuple
\[
  G^{*} = (V^{*},\; E_{\mathit{prog}}\cup\;E_{\mathit{ctx}}\cup\;
            E_{\mathit{fb}},\; N)
\]
where $V^{*} = V_{\mathit{prog}} \cup V_{\mathit{aux}}$ with
$V_{\mathit{aux}} \cap V_{\mathit{prog}} = \varnothing$;
\;$E_{\mathit{ctx}},\, E_{\mathit{fb}} \subseteq V^{*}\!\times V^{*}$
are the \emph{context-edge} and \emph{feedback-edge} relations; and
$N : V^{*} \to \mathcal{P}(V_{\mathit{prog}})$
is a \emph{code-neighborhood} function.
\end{definition}

Nodes in $V_{aux}$ are analysis nodes that do not correspond to a single program point; in the motivating example $V_{aux} = \{n_C\}$.  
A context edge $(m, n) \in E_{ctx}$ means that the current state (i.e., assessment of claims and evidence) at $m$ should be included when processing $n$.  
A feedback edge $(m, n) \in E_{fb}$ indicates that a change at $m$ should trigger re-evaluation of $n$.
We write $E_{ext} = E_{ctx} \cup E_{fb}$ to represent the \emph{extended edges}, and define predecessors and successors, restricted to nodes that are connected via extended edges:
\[
  \mathit{Pred}_{\mathit{ext}}(n)
    = \{m \mid (m,n)\in E_{\mathit{ext}}\},
  \qquad
  \mathit{Succ}_{\mathit{ext}}(n)
    = \{m \mid (n,m)\in E_{\mathit{ext}}\}.
\]

Program-derived edges do not drive worklist propagation unless they also appear in $E_{\mathit{ext}}$.
The code neighborhood $N(n)$ determines which source regions are visible when processing $n$.
For a program node $n \in V_{prog}$, $N(n)$ might be a k-hop control flow neighborhood, or simply the predecessors via $E_{prog}$.
For an auxiliary node $n \in V_{aux}$, $N(n)$ may span a larger region.
In the example, $N(n_C)$ includes the source of both \texttt{handleRequest} and \texttt{processPayload}.
We assume an implicit source map $\mathit{src}:V_{\mathit{prog}}\to\Sigma^{*}$ so that the prompt constructor can obtain code context for~$n$ by applying $\mathit{src}$ pointwise to $N(n)$.

\subsection{Claims, Assessment Domain, and State}

This subsection discusses what the states in our analysis look like.

\paragraph{Claims and Evidence.} A claim is a proposition to be assessed at a node.
For example, a claim could be a natural language assertion about program behavior, an interval assertion, or an API obligation.   
We let $\mathcal{C}$ denote the universe of claims used in the analysis, and we note that the core model is agnostic to the internal structure of the claim universe.
We let $\mathit{Ev}$ denote the universe of evidence records (presumably retrieved by the LLM agent).  For example, these can be documentation excerpts, advisory identifiers, source-code observations or tool outputs.
Evidence is accumulated alongside the assessments for provenance and auditability, but does not govern termination of the worklist procedure.

The assessment domain defines the space in which the framework summarizes its judgments about claims, given the evidence available to it.  More formally:
\begin{definition}[Assessment Domain]
    An \emph{assessment domain} is a finite-height join-semilattice
\[
\mathcal{A}
=
(A,\sqsubseteq,\sqcup,\bot_A).
\]
where each element of \(A\) summarizes the state of the evidence associated with a given claim.
\end{definition}

The $A_{\textsc{Graded}}$ lattice used in the motivating example is one concrete instantiation.
The finite-height requirement ensures that assessments can only be updated finitely many times, which is essential for termination.

The global state records the claims at each node, together with their current assessments and accumulated evidence, which are maintained by the LLM agent (e.g., using a scratchpad). More formally:

\begin{definition}[Global state]
The set $\mathbf{State}$ consists of maps
\[
  S : V^{*} \to
    \bigl(\mathcal{C} \rightharpoonup_{\!\mathit{fin}}
      (A \times \mathcal{P}_{\!\mathit{fin}}(\mathit{Ev}))\bigr)
\]
When $S(n)(c)=(a,E)$, claim $c$ at node $n$ currently has assessment~$a$ and accumulated evidence~$E$.
$\rightharpoonup_{\!\mathit{fin}}$ denotes that the domain of the partial map is finite, and $\mathcal{P}_{\!\mathit{fin}}$ denotes the finite subsets.
\end{definition}

The map \(S(n)\) is partial because each node carries only the claims that have been assigned to it during graph construction or generated at it during analysis.
When a new claim is generated at \(n\), it is initially inserted with assessment \(\bot_A\) and empty evidence before being evaluated.
We discuss claim generation, which did not appear in the motivating example, in the next subsection.
Although evidence is part of the state and may be made available to downstream or feedback-dependent nodes, the worklist trigger ignores evidence-only changes. The algorithm therefore stabilizes the active claims and their assessments, not the set of all evidence.

\subsection{Queries}

The core framework supports two kinds of queries.
\emph{Evaluation} queries serve as a rubric that tells the agent how to assess a claim (e.g., the bilateral query in Section~\ref{sec:context-queries-evidence}).
\emph{Generative} queries allow the LLM agent to generate candidate claims that are appropriate to the analysis problem at hand.
More formally, we say that an evaluation query $q_e \in \mathcal{Q}_e$ specifies how evidence should be sought, classified, and mapped to an element of~$\mathcal{A}$, where $\mathcal{Q}_e$ refers to the domain of possible evaluation queries.
A generative query $q_g \in \mathcal{Q}_g$ asks the agent to propose new claims relevant to the analysis, where $\mathcal{Q}_g$ refers to the domain of possible generative queries.
For each node $n$, a \emph{query specification} is a pair
\[
  \psi_n = (q_e,\; Q_g)
    \;\in\; \mathcal{Q}_e \times \mathcal{P}_{\!\mathit{fin}}(\mathcal{Q}_g),
\]
where $q_e$ is the evaluation rubric used at~$n$ and $Q_g$ is a finite set of generative queries.
We assume a map $\mathit{queries}: V^{*}\!\to \mathcal{Q}_e \times \mathcal{P}_{\!\mathit{fin}}(\mathcal{Q}_g)$ assigning specifications to nodes.
We use node-local queries for simplicity of presentation and to maintain generality.

\subsection{Contextualization and Node Transformers}
\label{sec:context-node-transformers}

This subsection defines how an individual node is processed. The role is analogous to a transfer function in abstract interpretation.
The first step is what we refer to as contextualization.

\paragraph{Contextualization.}
The function
\(
  \mathrm{ctx} : V^{*}\!\times\mathit{State}\to\Sigma^{*}
\)
assembles the prompt context for processing a node.  Given node $n$ and state~$S$,
\[
  \mathrm{ctx}(n,S)
  \;=\;
  \delta_n\!\bigl(
    \{\mathit{src}(m) \mid m \in N(n)\},\;
    \mathit{Goal},\;
    \{S(m) \mid m \in \mathit{Pred}_{\mathit{ext}}(n)\}
  \bigr),
\]
where $\mathit{Goal}$ is the analyst-provided high-level analysis goal and $\delta_n$ is a node-specific prompt constructor.  
This context assembly plays a role analogous to collecting information along incoming edges in chaotic iterations, though the aggregation here is prompt construction rather than a lattice join.
In practice, \(\delta_n\) may be a fixed template, a retrieval pipeline, an LLM
call, or a combination.

\paragraph{Agent interface.}
At each node, the agent exposes two operations:
\begin{align*}
  \mathrm{Gen}  &: \Sigma^{*}\times \mathcal{Q}_g \;\to\;
                    \mathcal{P}_{\!\mathit{fin}}(\mathcal{C}), \\
  \mathrm{Eval} &: \Sigma^{*}\times \mathcal{Q}_e \times \mathcal{C} \;\to\;
                    \mathcal{A} \times \mathcal{P}_{\!\mathit{fin}}(\mathit{Ev}).
\end{align*}

\(\mathrm{Gen}\) produces a finite set of candidate claims given a context string and a generative query. 
\(\mathrm{Eval}\) assesses an existing claim and returns both an assessment and a finite set of evidence records. For bilateral querying, as in the motivating example, \(\mathrm{Eval}\) may internally issue separate prompts for supporting and refuting evidence and then combine the results into a product-lattice element.
We assume that each call to the LLM terminates and returns a finite output.
If either operation returns malformed or uninterpretable responses, this will have to be addressed separately.

\paragraph{Node transformer.} The node transformer processes the node using $\mathrm{ctx}$, $\mathrm{Eval}$, and $\mathrm{Gen}$.
For each \(n \in V^*\), the node transformer $T_n:\mathbf{State}\to\mathbf{State}$ processes node $n$ as follows.
Given input state~$S$:
\begin{enumerate}
\item \textbf{Build context.}\;
  Set $\Gamma_n \leftarrow \mathit{ctx}(n,S)$ and retrieve
  $(q^e_n, Q^g_n) \leftarrow \mathit{queries}(n)$.
  Let $S'$ be a copy of~$S$.
 
\item \textbf{Evaluate existing claims.}\;
  For each $c\in\mathrm{dom}(S(n))$ with $S'(n)(c)=(a,E)$, compute
  $(a',E')\leftarrow\mathrm{Eval}(\Gamma_n,\,q^e_n,\,c)$
  and set $S'(n)(c)\leftarrow(a\sqcup a',\; E\cup E')$.
 
\item \textbf{Generate new claims.}\;
  For each $q_g\in Q^g_n$, compute
  $C_q\leftarrow\mathrm{Gen}(\Gamma_n,\,q_g)$.
  For each $c'\in C_q\setminus\mathrm{dom}(S'(n))$ accepted by the
  node-local claim-management policy (e.g., a policy that bounds the number of claims at a node and discards new claims if the limit is met), insert
  $S'(n)(c')\leftarrow(\bot_A,\,\varnothing)$.
 
\item \textbf{Evaluate new claims.}\;
  For each newly inserted $c'$, compute
  $(a',E')\leftarrow\mathit{Eval}(\Gamma_n,\,q^e_n,\,c')$
  and set $S'(n)(c')\leftarrow(a',\,E')$.
 
\item \textbf{Return} $S'$.
\end{enumerate}

\noindent
The transformer updates only the state at node \(n\); for all \(m \neq n\), \(S'(m)=S(m)\).
The join-based update ensures that assessments can only increase in the lattice order. Evidence is accumulated by union, but evidence-only changes do not add nodes to the worklist, as described below.

\subsection{Worklist Algorithm}
\label{sec:worklist}

The worklist algorithm (Algorithm~\ref{alg:worklist}) controls the global agentic interpretation procedure.
Propagation is driven by changes in the claims (i.e., assessments or the set of claims at a node) over the extended edges.
For a node-local state $M:\mathcal{C}\rightharpoonup_{\!\mathit{fin}}(\mathcal{A}\times\mathcal{P}_{\!\mathit{fin}}(\mathit{Ev}))$, we define the \emph{claim-assessment projection}
\[
  \mathit{AC}(M) : \mathcal{C}\rightharpoonup_{\!\mathit{fin}} \mathcal{A},
  \qquad
  \mathit{AC}(M)(c) = a
  \;\;\text{whenever}\;\; M(c)=(a,E).
\]

\begin{algorithm}[t]
\KwIn{Evaluation graph $G^*$, node transformers $\{T_n\}_{n\in V^*}$,
      initial state $S_0$.}
\KwOut{Stabilized state~$S$.}
$S \leftarrow S_0$\;
$W \leftarrow \{n \in V^* ~|~ |\mathrm{dom}(S_0(n))| > 0 \vee \text{$n$ has generative queries}\}$\;
\While{$W \neq \varnothing$}{
  remove some node $n$ from $W$\;
  $\mathit{old} \leftarrow \mathit{AC}(S(n))$\;
  $S \leftarrow T_n(S)$\;
  \If{$\mathit{AC}(S(n)) \neq \mathit{old}$}{
    $W \leftarrow W \cup \mathit{Succ}_{\mathit{ext}}(n)$\;
  }
}
\Return $S$\;
\caption{Worklist Harness}\label{alg:worklist}
\end{algorithm}

Algorithm~\ref{alg:worklist} terminates when the worklist is empty.
The returned state is stabilized with respect to claims and their assessments.
The returned state need not be evidence-stabilized in the stronger sense that no further evidence could ever be found, because evidence-only updates do not trigger reprocessing.

\begin{proposition}[Termination]\label{prop:term} Assume: (1)~$V^*$ is finite; (2)~$A$ has finite height~$H$; (3)~for each node $n$, at most $k_n$ claims can ever be active; (4)~claims, once introduced, are never revoked; (5)~each $T_n$ updates assessments only by join; and (6)~each call to $\mathit{Gen}$ and $\mathit{Eval}$ terminates with finite output. Then Algorithm~\ref{alg:worklist} terminates.
\end{proposition}

\begin{proof}
Let $K_{\mathit{cl}}=\sum_{n\in V^*} k_n$.
Since claims are never revoked, at most $K$ claims are ever active across the entire graph.  Each active claim's assessment can strictly increase at most $H$ times, giving at most
$H\!\cdot\! K_{\mathit{cl}}$ strict assessment increases in total.
The algorithm adds successors to the worklist only when $\mathit{AC}(S(n))$ changes, which requires either a new claim or a strict increase. 
Hence, there are at most $K_{\mathit{cl}} + H\!\cdot\! K_{\mathit{cl}}$ worklist-triggering events, each enqueueing finitely many nodes.
Combined with the finite initial worklist and the termination of each agent call, the worklist eventually empties, and the algorithm terminates.
\end{proof}

In practice, the bound \(k_n\) can be enforced by having \(T_n\) discard newly generated claims once the node-local limit has been reached, or choose which to keep. This process itself could be governed by an LLM agent.
Proposition~\ref{prop:term} proves finite stabilization of the claim-and-assessment projection, but it does not assert a least fixed point.
That would require deterministic, monotone node transformers and fair scheduling, which are conditions current LLM agents do not satisfy in general.
Techniques such as fixed prompts, greedy decoding, and memoization can reduce nondeterminism but do not eliminate it; we therefore state only termination here, and leave exploration of this direction for future work.
\section{Facets of the Design}
The core model introduced in Section~\ref{sec:core-model} is parametric.
Each component: the assessment lattice, the queries, claims, evidence, worklist policy, and the evaluation graph itself admits a wide range of design choices that affect the cost, expressiveness, and interpretability of the analysis.
This section discusses the main design choices in more detail.

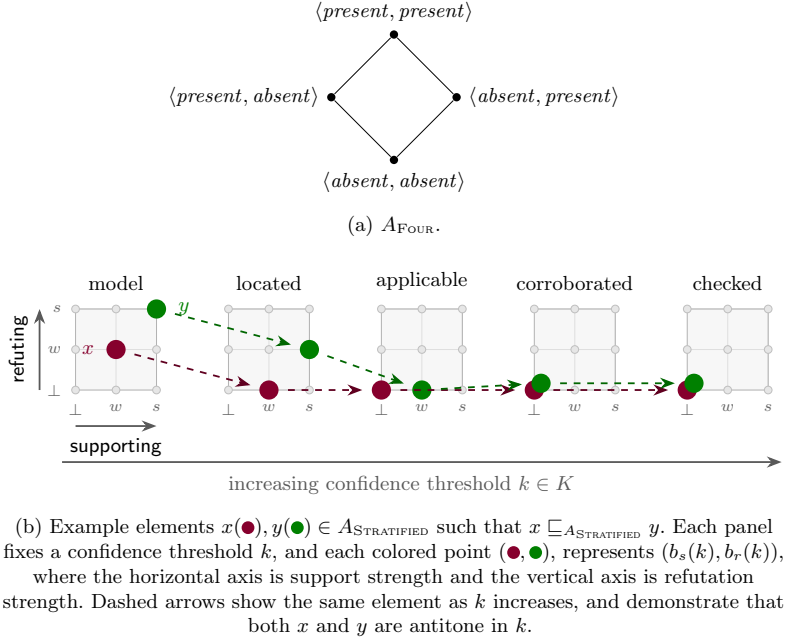
\begin{figure}[t]
\centering
\scalebox{0.85}{
\begin{minipage}{\textwidth}
\centering
 
\begin{tikzpicture}[
  scale=0.65,
  every node/.style={font=\footnotesize},
  dot/.style={circle, fill, inner sep=1.3pt}
]
  \node[dot, label=below:{$\langle \mathit{absent}, \mathit{absent} \rangle$}] (bot) at (0, 0.5) {};
  \node[dot, label=left:{$\langle \mathit{present}, \mathit{absent} \rangle$}]  (sup) at (-1.5, 2) {};
  \node[dot, label=right:{$\langle \mathit{absent}, \mathit{present} \rangle$}] (ref) at (1.5, 2) {};
  \node[dot, label=above:{$\langle \mathit{present}, \mathit{present} \rangle$}] (top) at (0, 3.5) {};
  \draw (bot) -- (sup);
  \draw (bot) -- (ref);
  \draw (sup) -- (top);
  \draw (ref) -- (top);
\end{tikzpicture}
 
\smallskip
{\small (a) $A_{\textsc{Four}}$.}
 
\bigskip
 
\begin{tikzpicture}[
    scale=0.7,
    >=Stealth,
    gridnode/.style={circle, inner sep=1.2pt, fill=gray!20, draw=gray!50},
    purpledot/.style={circle, inner sep=3pt, fill=purple!70!black},
    greendot/.style={circle, inner sep=3pt, fill=green!50!black},
  ]
 
  \def\panelgap{3.4}
  \def\gs{0.9}
 
  \foreach \klabel/\px/\py/\gx/\gy [count=\i from 0] in {
    {model}/1/1/2/2,
    {located}/1/0/2/1,
    {applicable}/0/0/1/0,
    {corroborated}/0/0/0/0,
    {checked}/0/0/0/0%
  } {
    \begin{scope}[xshift=\i*\panelgap cm]
 
      \begin{scope}[on background layer]
        \fill[gray!6] (0,0) rectangle (2*\gs, 2*\gs);
      \end{scope}
 
      \draw[gray!25, thin] (\gs,0) -- (\gs,2*\gs);
      \draw[gray!25, thin] (2*\gs,0) -- (2*\gs,2*\gs);
      \draw[gray!25, thin] (0,\gs) -- (2*\gs,\gs);
      \draw[gray!25, thin] (0,2*\gs) -- (2*\gs,2*\gs);
 
      \draw[gray!60] (0,0) rectangle (2*\gs, 2*\gs);
 
      \foreach \x in {0,1,2} {
        \foreach \y in {0,1,2} {
          \node[gridnode] at (\x*\gs, \y*\gs) {};
        }
      }
 
      \node[below, font=\scriptsize, gray!80!black] at (0, -0.15) {$\bot$};
      \node[below, font=\scriptsize, gray!80!black] at (\gs, -0.15) {$w$};
      \node[below, font=\scriptsize, gray!80!black] at (2*\gs, -0.15) {$s$};
 
      \ifnum\i=0
        \node[left, font=\scriptsize, gray!80!black] at (-0.15, 0) {$\bot$};
        \node[left, font=\scriptsize, gray!80!black] at (-0.15, \gs) {$w$};
        \node[left, font=\scriptsize, gray!80!black] at (-0.15, 2*\gs) {$s$};
 
        \draw[->, gray!70!black, thick] (-0.8, 0) -- (-0.8, 2*\gs);
        \node[left, font=\small, rotate=90, anchor=south] at (-0.9, \gs) {\textsf{refuting}};
 
        \draw[->, gray!70!black, thick] (0, -0.8) -- (2*\gs, -0.8);
        \node[below, font=\small] at (\gs, -0.9) {\textsf{supporting}};
      \fi
 
      \node[above, font=\small] at (\gs, 2*\gs + 0.25) {\klabel};
 
      \node[purpledot] (p\i) at (\px*\gs, \py*\gs) {};
      \pgfmathtruncatemacro{\overlap}{\px==\gx && \py==\gy ? 1 : 0}
      \ifnum\overlap=1
        \node[greendot] (g\i) at (\gx*\gs + 0.15, \gy*\gs + 0.15) {};
      \else
        \node[greendot] (g\i) at (\gx*\gs, \gy*\gs) {};
      \fi
 
      \ifnum\i=0
        \node[left=2pt of p\i, font=\small\bfseries, purple!70!black] {$x$};
        \node[right=2pt of g\i, font=\small\bfseries, green!50!black] {$y$};
      \fi
 
    \end{scope}
  }

  \foreach \a/\b in {0/1, 1/2, 2/3, 3/4} {
    \draw[->, purple!50!black, thick, dashed, shorten >=4pt, shorten <=4pt]
      (p\a) -- (p\b);
  }
  \foreach \a/\b in {0/1, 1/2, 2/3, 3/4} {
    \draw[->, green!40!black, thick, dashed, shorten >=4pt, shorten <=4pt]
      (g\a) -- (g\b);
  }
 
  \draw[->, thick, gray!70!black] (-0.3, -1.6) -- (4*\panelgap + 2*\gs + 0.3, -1.6);
  \node[below, font=\small, gray!70!black] at (2*\panelgap + \gs*0.5, -1.75) {%
    increasing confidence threshold $k \in K$};
 
  \path (4*\panelgap + 2*\gs + 1.4, 0);
\end{tikzpicture}
 
\smallskip
{\small (b) Example elements $x (\purpmark), y (\greenmark) \in A_{\textsc{Stratified}}$ such that $x \sqsubseteq_{A_{\textsc{Stratified}}} y$.
Each panel fixes a confidence threshold $k$, and each colored point $(\purpmark, \greenmark)$, represents $(b_s(k), b_r(k))$, where the horizontal axis is support strength and the vertical axis is refutation strength.
Dashed arrows show the same element as $k$ increases, and demonstrate that both $x$ and $y$ are antitone in $k$.
}
 
\end{minipage}%
}
 
\caption{Alternative assessment domains.}

\label{fig:assessment-domains}
\end{figure}

\subsection{Choosing the Assessment Lattice}
\label{sec:choosing-lattice}

The assessment lattice determines how many distinct judgments the framework can express about a claim.
A coarser lattice is cheaper to operate and easier to interpret, but a finer lattice preserves more distinctions and requires the LLM agent to make granular judgments reliably.
A simple bilateral option is a two-bit evidence-presence lattice ($A_{\textsc{Four}}, $ Figure~\ref{fig:assessment-domains}a): one bit records whether supporting evidence has been observed, and the other records whether refuting evidence has been observed.
This domain suffices when the analysis only needs to know whether each polarity has appeared, but it cannot distinguish weak evidence from strong evidence.  Here, polarity refers to supporting or refuting evidence.
$A_{\textsc{Graded}}$ (Section~\ref{sec:graded-domain}) refines presence to three levels per polarity ($\bot \sqsubseteq w \sqsubseteq s$), distinguishing absent, weak, and strong evidence on each side.
This is expressive enough when all evidence is of comparable reliability, but it does not track \emph{why} the evidence should be trusted. A claim backed by an LLM-only judgment and a claim backed by a human-audited finding could both register as strong support in $A_{\textsc{Graded}}$.
To recover this distinction, we introduce a finite total order over \emph{confidence-bases}, $K$:
\[
  \mathit{model} \;\sqsubseteq_K\; \mathit{located} \;\sqsubseteq_K\; \mathit{applicable} \;\sqsubseteq_K\; \mathit{corroborated} \;\sqsubseteq_K\; \mathit{checked}.
\]
Here, \emph{model} denotes an LLM-only judgment with no independently inspectable evidence; \emph{located} denotes a judgment backed by a concrete evidence record; \emph{applicable} denotes evidence confirmed to apply to the specifics of the program under analysis (e.g., a specific version of a library); \emph{corroborated} denotes agreement among independent applicable sources; and \emph{checked} denotes validation by source code inspection, a deterministic tool, verifier, or human audit.
$k \sqsubseteq_K k'$ means that evidence with basis $k'$ is considered at least as trustworthy as $k$.

A na\"{\i}ve product $A_{\textsc{Graded}} \times K$ attaches one basis value to the entire support-refutation pair, but this can be misleading.
Strong $\mathit{model}$ support paired with weak $\mathit{checked}$ refutation would make the whole assessment appear ``$\mathit{checked}$.''  
A more precise domain should instead track the basis of each polarity separately.
In what follows, let $G = \{\bot, w, s\}$ with $\bot \sqsubseteq_G w \sqsubseteq_G s$ denote the evidence-strength chain used in $A_{\textsc{Graded}}$.
Tracking a single $(g, k) \in G \times K$ per polarity does not resolve the issue: joining $(s, \mathit{model})$ with $(w, \mathit{checked})$ (e.g., in a component-wise join) yields $(s, \mathit{checked})$, falsely attributing strong evidence to a checked source.

We thus introduce the domain $A_{\textsc{Stratified}}$.
Suppose the framework has gathered a finite set of evidence records bearing on one polarity of a claim, each carrying a strength $g_i \in G$ and a confidence basis $k_i \in K$.
The \emph{basis-indexed summary} for that polarity is the map $b : K \to G$ defined by
\[
  b(k) \;=\; \bigsqcup_{\{i \,:\, k \,\sqsubseteq_K\, k_i\}} g_i,
\]
where the empty join is equal to~$\bot_G$.
Thus, $b(k)$ is the strongest evidence available after discarding all records whose confidence basis falls below threshold $k$.
Raising the confidence threshold can only remove records, so $b$ is antitone ($k \sqsubseteq_K k' \implies b(k') \sqsubseteq_G b(k)$).
Let $K \to_{\downarrow} G$ denote the join-semilattice of antitone maps from~$K$ to~$G$, ordered pointwise, with join $(b \sqcup_{K \to_{\downarrow} G} b') (k) = b(k) \sqcup_G b'(k)$.
The \emph{stratified assessment} of a single claim is a pair of such maps, one for each polarity:
\[
  A_{\textsc{Stratified}} \;=\;
    \underbrace{(K \to_{\downarrow} G)}_{\text{support}} \;\times\;
    \underbrace{(K \to_{\downarrow} G)}_{\text{refutation}},
\]
also ordered and joined pointwise.
The assessment of a claim $c$ in $A_{\textsc{Stratified}}$ records, separately for refutation and support, the strongest evidence available at each confidence threshold.
The partial order is defined element-wise: $\langle b_s, b_r \rangle \sqsubseteq_{A_{\textsc{Stratified}}} \langle b_s', b_r'\rangle \iff \forall k \in K ~.~ b_s(k) \sqsubseteq_G b'_s(k) \wedge b_r(k) \sqsubseteq_G b'_r(k)$.  An example of two such elements is shown in Figure~\ref{fig:assessment-domains}b.
Join is defined as $\langle b_s, b_r \rangle \sqcup_{\mathcal{A}_{\textsc{Stratified}}} \langle b_s', b_r'\rangle = \langle b_s \sqcup_{K \to_{\downarrow} G} b_s', b_r \sqcup_{K \to_{\downarrow} G} b_r' \rangle$.
$\bot_{A_{\textsc{Stratified}}}$ is defined as $\langle \bot_{K \to_{\downarrow} G}, \bot_{K \to_{\downarrow} G} \rangle$, where $\bot_{K \to_{\downarrow} G}(k) = \bot_G$ for all $k \in K$.
We note that antitonicity is a property of the individual elements, not of the analysis dynamics.
A claim's assessment still evolves monotonically throughout the analysis, via polarity-wise join. 
Furthermore, we note that the previous issue of collapsing $(s, \mathit{model})$ and $(w, \mathit{checked})$ into $(s, \mathit{checked})$ does not occur.  Instead, for example, for supporting evidence $\mathcal{A}_{\textsc{Stratified}}$ represents these concepts separately: $b_s(\mathit{model}) = s$ and $b_s(\mathit{checked}) = w$.

\subsection{Claims}

The core model leaves the claim universe $\mathcal{C}$ unstructured.
Three considerations could guide instantiation.

\paragraph{Ordered claim families.}
Claim contents may inhabit an ordered domain independent of the assessment lattice, such as interval inclusion for numeric bounds or set inclusion for taint sources.  Such structure enables detecting subsumption among generated claims or replacing related claims with a coarser representative.

\paragraph{Canonicalization.}
Claims generated in natural language may carry the same information in many different forms.
Without normalization, the harness could treat equivalent claims as distinct, duplicating work and fragmenting evidence.
Implementations should canonicalize claims before insertion (e.g., via templates or an LLM judge), merging them when they have the same semantic meaning.

\paragraph{Granularity and scope.}
Claims that are too coarse (e.g., ``the verifier is safe'') conflate distinct obligations and resist focused evaluation, while claims that are too fine proliferate near-redundant obligations and inflate worklist cost.  
In practice, this task is dependent on how well an LLM agent is able to synthesize claims, in terms of diversity and scope and how well it is able to deconstruct high-level goals into well-scoped claims.

\subsection{Query Design}

The evaluation query determines how the agent produces the judgments recorded in the lattice.  For any bilateral lattice, it naturally decomposes into a support subquery and a refutation subquery (Section~\ref{sec:context-queries-evidence}).
Whether these are issued in a single prompt or separately is an implementation choice that trades LLM calls against clarity and potential anchoring effects~\cite{DBLP:journals/corr/abs-2511-05766,DBLP:journals/expert/OLeary25a,DBLP:journals/jocss/LouS26}.
With the stratified domain $A_{\textsc{Stratified}}$, the query should ask the agent to return structured evidence records, where each has a polarity, strength grade, and confidence basis, instead of a single lattice element.
The harness could then aggregate them via pointwise join as described in Section~\ref{sec:context-node-transformers}.
For simpler lattices such as $A_{\textsc{Graded}}$, the agent can report strength directly.
On the other hand, generative queries ask the agent to propose new claims.  Because the per-node claim bound (Section~\ref{sec:worklist}) enforces termination, the query should signal parsimony: only claims genuinely relevant to the analysis goal.
An example of a generative query could be ``Generate a claim that states what interval range variable $x$ should be in.''
The dynamics of such queries, how their results could be updated (e.g., with standard interval joins) over time, and how the generative queries might have to change throughout the worklist procedure, is something we leave to future work to explore.

\subsection{Evidence and Provenance}

The lattice captures a compact judgment, but for the result to be inspectable, the evidence behind each judgment must be retained.
Evidence records could carry the structure defined in Section~\ref{sec:choosing-lattice} (i.e., polarity, strength, and (for $A_{\textsc{Stratified}}$) confidence basis) along with source material (e.g., documentation excerpt, advisory identifier, code observation, or tool output) and the claim and node to which they pertain.
Instantiations that support belief revision (Section~\ref{sec:revision}) should additionally track whether each record is active, superseded, or retracted.
In Section~\ref{sec:core-model}, we assumed that evidence-only updates do not trigger worklist propagation.
However, in an implementation or an alternative formal model we may want to consider this to account for new evidence that may not change the evaluation of the current claims at a node, but may impact others.

\subsection{Worklist Order}

In classical dataflow analysis, worklist order affects performance but not the final result.
Agentic interpretation does not enjoy this guarantee: the LLM's response may depend on the context available at query time, making worklist order a design choice that can affect result quality, not just cost.
Several policies merit exploration.
A \emph{weak topological ordering} policy could be used~\cite{DBLP:conf/ershov/Bourdoncle93}.
A \emph{goal-directed} policy begins at a goal node and adds predecessors as needed to the worklist.
A \emph{feedback-sensitive} policy prioritizes recently refuted nodes, since refutation may indicate that upstream investigation should be narrowed.
The ordering could itself be delegated to an LLM agent.

\subsection{Assessment and Claim Revision}
\label{sec:revision}

Algorithm~\ref{alg:worklist} moves assessments only upward via join.
Two situations motivate downward revision: (1)~later evidence may reveal that an earlier judgment was based on inapplicable or misattributed evidence, and the corresponding assessment component should be lowerable; (2)~continued investigation may show that a claim was poorly scoped or redundant, warranting its removal.  
Both forms break the inflationary invariant on which Proposition~\ref{prop:term} relies.
Two disciplined alternatives preserve termination:
 
\paragraph{Bounded replacement.}
Each node is allowed at most $I_{\max}$, $R_{\max}$, and $D_{\max}$ claim introductions, claim retractions, and downward assessment moves, respectively. 
Between revisions, the standard join-based update applies.
Termination follows because finitely many revision events can occur, each re-enabling at most~$H$ subsequent upward steps.
 
\paragraph{Epochal recomputation.}
Revision is separated from accumulation.
When the worklist empties and the current epoch stabilizes, the harness identifies assessments to lower and claims to retract, records old values with the reason for revision, resets affected assessments to~$\bot_A$, removes retracted claims, and re-seeds the worklist for a new epoch.  Within each epoch, Proposition~\ref{prop:term} applies directly.
Termination of the multi-epoch procedure requires a separate mechanism such as a global epoch limit. \\

\noindent
In recent work, Jenkins~\cite{jenkins2026agmbench} demonstrates that frontier LLMs struggle to revise their beliefs rationally.
In future work, we intend to explore whether agentic interpretation can ameliorate some of the issues described by Jenkins.

\subsection{Evaluation Graph Synthesis}

The preceding subsections assume that the evaluation graph is given. 
In practice, deciding which auxiliary nodes to introduce, which claims and queries to assign, and which extended edges to draw is itself a significant design task.
An LLM could propose auxiliary nodes, claims, and extended edges from the high-level goal, with the harness bounding and auditing the result (e.g., limiting the number of auxiliary nodes and outgoing extended edges per node).
Dynamic graph refinement during the worklist computation is also possible, but the effect on termination must be carefully managed: if the graph can grow during analysis, the termination argument of Proposition~\ref{prop:term} must be extended to account for the additional claims and edges that each refinement step introduces.
\section{Related Work}

Agentic interpretation is most related to work that uses LLMs inside program-analysis procedures.
We organize the discussion in this section around the role assigned to the LLM in such systems.

\paragraph{LLMs as generators for checkable artifacts.}
Many LLM-assisted verification works follow a generate-and-check pattern.
Examples include invariant generation~\cite{DBLP:conf/icml/PeiBSSY23,wei2026quokka}, proof generation and repair~\cite{DBLP:conf/sigsoft/FirstRRB23}, postcondition generation~\cite{DBLP:journals/pacmse/EndresFCL24}, formal contract inference~\cite{DBLP:journals/corr/abs-2510-12702}, and combining static analysis and deductive verification with LLMs to generate and refine formal specifications~\cite{DBLP:journals/corr/abs-2508-14532}.
These systems have correctness guarantees; the LLM's output is consumed by a formal checker or validity is enforced by the system.
Agentic interpretation instead targets evidence-dependent analyses in which the relevant facts may come from sources that are difficult to formalize and where no complete verifier exists.
In this setting, the framework cannot make the LLM's conclusion sound; instead, it makes the LLM's evidential role explicit, bounded, and auditable.

\paragraph{LLM-assisted static analysis.}
Recent works use LLMs inside static-analysis procedures to supply local judgments, specifications, or domain knowledge that would otherwise require manual engineering.
LLift uses an LLM to answer localized path-feasibility questions when triaging use-before-initialization reports~\cite{DBLP:journals/pacmpl/LiHZQ24}.
IRIS uses LLMs to infer taint specifications for CodeQL, allowing the analyzer to reason about library behavior not captured by built-in rules~\cite{DBLP:conf/iclr/Li0N25}. 
NESA leverages LLM-derived semantic relations and parser-derived syntactic facts within a Datalog-style policy framework to enable customizable, compilation-free static program analysis~\cite{wang2026nesa}, while LATTE constructs LLM prompt sequences for binary taint analysis from sliced code contexts~\cite{DBLP:journals/corr/abs-2310-08275}. 
E\&V prompts LLMs to simulate pseudocode execution and check intermediate analysis steps~\cite{DBLP:journals/corr/abs-2312-08477}.
Chapman et al. interleave static-analysis intermediate results with LLM prompts for error-specification inference~\cite{DBLP:journals/sttt/ChapmanRT25}.
ZeroFalse uses LLMs to adjudicate SAST warnings using structured traces and CWE-specific context~\cite{DBLP:journals/corr/abs-2510-02534}.
AbsInt-AI uses language models to improve the precision of an abstract interpreter while retaining the guarantees of the underlying analyzer~\cite{wang2025absintai}.
These systems are the closest predecessors of agentic interpretation.
Specifically, they leverage LLMs to answer local program-analysis questions that may be difficult to answer otherwise.
The key difference between these lines of work and agentic interpretation is that each system mentioned fixes a particular analysis task and a particular role for the LLM.
Instead, agentic interpretation abstracts the surrounding control structure, and focuses on formal scaffolding to guide LLM-based reasoning.
Soundness is not necessarily guaranteed in the framework; instead, the focus is on turning many local LLM judgments into a bounded, inspectable analysis state.

\paragraph{Structured agentic code reasoning.}
General structured-reasoning frameworks also inform our design, but they solve a different problem.
ReAct interleaves reasoning traces with tool actions~\cite{DBLP:conf/iclr/YaoZYDSN023}, and Graph of Thoughts represents intermediate reasoning units as a graph with feedback~\cite{DBLP:conf/aaai/BestaBKGPGGLNNH24}.
More generally, these frameworks organize how an LLM explores a problem, whereas agentic interpretation also organizes what the analysis has established and focuses on how the analysis state stabilizes w.r.t. a finite assessment domain.
Ugare and Chandra's agentic code reasoning is closer to our setting. Specifically, their semi-formal certificates require agents to state premises, trace execution paths, and derive conclusions before committing to a code-analysis judgment~\cite{DBLP:journals/corr/abs-2603-01896}.
Agentic interpretation is complementary.
Semi-formal certificates constrain the structure of a single agent’s reasoning trace for a task instance, while agentic interpretation constrains the system-level evolution of claim judgments.
That is, how many claims are active, how their assessments are updated by join, how evidence records are retained, and how assessment changes propagate across dependent program claims.
We believe that combining the two directions could lead to interesting future work.

\paragraph{Principled neurosymbolic reasoning.}
Recent neurosymbolic frameworks aim to replace ad hoc chains of LLM calls with models that preserve some discipline from symbolic reasoning.
Bembenek proposes Neurosymbolic Transition Systems, in which symbolic state is paired with neural ``intuition'' so that an LLM can guide a symbolic algorithm while the overall computation inherits properties of the symbolic transition system~\cite{DBLP:journals/corr/abs-2507-05886}.
Allen et al. integrate an LLM into the interpretation function of a paraconsistent logic, issuing separate verification and refutation queries for atomic formulas while preserving the soundness and completeness of the surrounding tableau calculus~\cite{DBLP:conf/nesy/AllenCFIG25}.
The difference between these works and agentic interpretation is that these frameworks begin with a symbolic calculus or transition system and insert LLM calls into it.
Agentic interpretation begins with an evidence-dependent program analysis goal for which no complete symbolic decision procedure is assumed.
Moreover, the soundness of these systems follows from the underlying soundness of the symbolic substrate upon which they are built.
Agentic interpretation assumes no such substrate and instead focuses on developing a lattice-based backbone for LLM code reasoning.

\paragraph{Abstract interpretation, bilattices, and revision.}
The formal control structure of agentic interpretation borrows from classical abstract interpretation and dataflow analysis~\cite{DBLP:conf/popl/CousotC77,DBLP:conf/popl/CousotC79,DBLP:conf/popl/Kildall73}. 
In classical abstract interpretation, the lattice element denotes an abstraction of concrete program states and supports a soundness theorem relating abstract and concrete semantics.
In agentic interpretation, the lattice element denotes the current evidential status of a claim. Stabilization therefore means that the process has reached a bounded stabilized point of recorded LLM assessments, not that it has computed a sound over-approximation of program behavior.
The bilateral assessment domains used in this paper are related to Belnap-style and bilattice semantics, where support for truth and support for falsity can be represented independently~\cite{DBLP:journals/jlp/Fitting91}.
Related multi-valued ideas have been used in program analysis, including three-valued shape analysis and multi-valued model checking~\cite{DBLP:journals/toplas/SagivRW02,DBLP:journals/fmsd/ChechikGDLE06}.
However, our use is epistemic rather than semantic. That is, a value such as strong support and weak refutation describes what evidence has been observed for a claim, not that the underlying program property is simultaneously true and false.
Finally, the controlled revision mechanisms in Section~\ref{sec:revision} are related to AGM belief revision and truth-maintenance systems, which track justifications for assertions that may later be retracted~\cite{DBLP:journals/ai/Doyle79,DBLP:journals/ai/Kleer86,DBLP:journals/jsyml/AlchourronGM85}.
\section{Limitations, Future Work, and Open Questions}

This paper presents a formal framework and a detailed worked example, but does not include an implementation or evaluation.
Building a system that instantiates the framework is the most immediate next step.
Beyond implementation, three directions to investigate stand out.
First, the framework structures the analysis and guarantees termination (under certain conditions), but cannot compensate for an agent that consistently misjudges evidence; it makes such errors visible rather than correcting them.
Second, programs may contain regions where sound analysis is available, and a natural direction is to use formal tools at nodes where they apply and LLM-backed queries where they do not, though ensuring coherent interaction between formally verified and LLM-assessed evidence requires further investigation.
Third, how the framework scales to larger programs with many more components and claims remains to be investigated.
\vspace{-5pt}
\section{Conclusion}
This paper introduces agentic interpretation, a framework that brings the structural discipline of lattice-based static analysis to LLM-driven program reasoning.
We formalized the core model, including the evaluation graph, assessment domain, node transformers, and worklist algorithm, and proved termination for finite-height lattice instantiations with bounded claim generation.
We further characterized the design space: the choice of assessment domain, claim representation and canonicalization, query design, worklist ordering, evidence provenance, and controlled forms of belief revision.
Our hope is that this paper inspires future discussions about designing formally-grounded techniques for agentic systems in the context of program analysis that go beyond generate-and-check and apply in settings where traditional verification methods are inapplicable.

\bibliographystyle{splncs04}
\bibliography{references}
\end{document}